\documentclass{bmcart}

\usepackage[utf8]{inputenc} %unicode support

\usepackage{amsfonts}    % Enable this line and disable the
\usepackage{amsmath}
\usepackage{amssymb}
\usepackage{graphicx,color}          % Include this line if your
\usepackage[dvips]{epsfig}    % or this line, depending on which
\DeclareMathOperator{\tr}{Tr}

\newtheorem{Def}{Definition}

\newtheorem{prop}{Proposition}

\newtheorem{proof}{Proof}

\def\includegraphics{}

\startlocaldefs
\endlocaldefs

\begin{document}

%%% Start of article front matter
\begin{frontmatter}

\begin{fmbox}
\dochead{Research}


\title{Interpolation Approach to Hamiltonian-varying Quantum Systems and the Adiabatic Theorem}

%%%%%%%%%%%%%%%%%%%%%%%%%%%%%%%%%%%%%%%%%%%%%%
%%                                          %%
%% Enter the authors here                   %%
%%                                          %%
%% Specify information, if available,       %%
%% in the form:                             %%
%%   <key>={<id1>,<id2>}                    %%
%%   <key>=                                 %%
%% Comment or delete the keys which are     %%
%% not used. Repeat \author command as much %%
%% as required.                             %%
%%                                          %%
%%%%%%%%%%%%%%%%%%%%%%%%%%%%%%%%%%%%%%%%%%%%%%

\author[
   addressref={aff1},                   % id's of addresses, e.g. {aff1,aff2}
   corref={aff1},                       % id of corresponding address, if any
   email={yu.pan.83.yp@gmail.com}   % email address
]{\inits{YP}\fnm{Yu} \snm{Pan}}
\author[
   addressref={aff2},
   email={Zibo.Miao@unimelb.edu.au}
]{\inits{ZM}\fnm{Zibo} \snm{Miao}}

\author[
   addressref={aff3},
   email={nina.amini@lss.supelec.fr}
]{\inits{NA}\fnm{Nina H.} \snm{Amini}}

\author[
   addressref={aff4},
   email={v.ugrinovskii@gmail.com}
]{\inits{VU}\fnm{Valery} \snm{Ugrinovskii}}

\author[
   addressref={aff1},
   email={matthew.james@anu.edu.au}
]{\inits{MJ}\fnm{Matthew R.} \snm{James}}

%%%%%%%%%%%%%%%%%%%%%%%%%%%%%%%%%%%%%%%%%%%%%%
%%                                          %%
%% Enter the authors' addresses here        %%
%%                                          %%
%% Repeat \address commands as much as      %%
%% required.                                %%
%%                                          %%
%%%%%%%%%%%%%%%%%%%%%%%%%%%%%%%%%%%%%%%%%%%%%%

\address[id=aff1]{%                           % unique id
  \orgname{Research School of Engineering, Australian National University}, % university, etc
  \postcode{0200},                                % post or zip code
  \city{Canberra},                              % city
  \cny{Australia}                                    % country
}

\address[id=aff2]{%
  \orgname{Department of Electrical \& Electronic Engineering, The University of Melbourne},
  \postcode{3010}
  \city{Melbourne VIC},
  \cny{Australia}
}

\address[id=aff3]{%
  \orgname{CNRS, Laboratoire des signaux et syst\`{e}mes (L2S) Sup\'{e}lec, 3 rue Joliot-Curie},
  \postcode{91192}
  \city{Gif-Sur-Yvette},
  \cny{France}
}
\address[id=aff4]{%
  \orgname{School of Engineering and Information Technology, University of New South Wales at ADFA},
  \postcode{2600}
  \city{Canberra},
  \cny{Australia}
}

%%%%%%%%%%%%%%%%%%%%%%%%%%%%%%%%%%%%%%%%%%%%%%
%%                                          %%
%% Enter short notes here                   %%
%%                                          %%
%% Short notes will be after addresses      %%
%% on first page.                           %%
%%                                          %%
%%%%%%%%%%%%%%%%%%%%%%%%%%%%%%%%%%%%%%%%%%%%%%

\begin{artnotes}
%\note{Sample of title note}     % note to the article
% note, connected to author
\end{artnotes}

\end{fmbox}% comment this for two column layout

%%%%%%%%%%%%%%%%%%%%%%%%%%%%%%%%%%%%%%%%%%%%%%
%%                                          %%
%% The Abstract begins here                 %%
%%                                          %%
%% Please refer to the Instructions for     %%
%% authors on http://www.biomedcentral.com  %%
%% and include the section headings         %%
%% accordingly for your article type.       %%
%%                                          %%
%%%%%%%%%%%%%%%%%%%%%%%%%%%%%%%%%%%%%%%%%%%%%%

\begin{abstractbox}

\begin{abstract} % abstract
Quantum control could be implemented by varying the system Hamiltonian. According to adiabatic theorem, a slowly changing Hamiltonian can approximately keep the system at the ground state during the evolution if the initial state is a ground state. In this paper we consider this process as an interpolation between the initial and final Hamiltonians. We use the mean value of a single operator to measure the distance between the final state and the ideal ground state. This measure resembles the excitation energy or excess work performed in thermodynamics, which can be taken as the error of adiabatic approximation. We prove that under certain conditions, this error can be estimated for an arbitrarily given interpolating function. This error estimation could be used as guideline to induce adiabatic evolution. According to our calculation, the adiabatic approximation error is not linearly proportional to the average speed of the variation of the system Hamiltonian and the inverse
of the energy gaps in many cases. In particular, we apply this analysis to an example in which the applicability of the adiabatic theorem is questionable.
\end{abstract}

%%%%%%%%%%%%%%%%%%%%%%%%%%%%%%%%%%%%%%%%%%%%%%
%%                                          %%
%% The keywords begin here                  %%
%%                                          %%
%% Put each keyword in separate \kwd{}.     %%
%%                                          %%
%%%%%%%%%%%%%%%%%%%%%%%%%%%%%%%%%%%%%%%%%%%%%%


% MSC classifications codes, if any
%\begin{keyword}[class=AMS]
%\kwd[Primary ]{}
%\kwd{}
%\kwd[; secondary ]{}
%\end{keyword}

\end{abstractbox}
%
%\end{fmbox}% uncomment this for twcolumn layout

\end{frontmatter}

%%%%%%%%%%%%%%%%%%%%%%%%%%%%%%%%%%%%%%%%%%%%%%
%%                                          %%
%% The Main Body begins here                %%
%%                                          %%
%% Please refer to the instructions for     %%
%% authors on:                              %%
%% http://www.biomedcentral.com/info/authors%%
%% and include the section headings         %%
%% accordingly for your article type.       %%
%%                                          %%
%% See the Results and Discussion section   %%
%% for details on how to create sub-sections%%
%%                                          %%
%% use \cite{...} to cite references        %%
%%  \cite{koon} and                         %%
%%  \cite{oreg,khar,zvai,xjon,schn,pond}    %%
%%  \nocite{smith,marg,hunn,advi,koha,mouse}%%
%%                                          %%
%%%%%%%%%%%%%%%%%%%%%%%%%%%%%%%%%%%%%%%%%%%%%%

%%%%%%%%%%%%%%%%%%%%%%%%% start of article main body
% <put your article body there>

%%%%%%%%%%%%%%%%
%% Background %%
%%
\section{Introduction}

Adiabatic process is aimed at stabilizing a parameter-varying quantum
system at its eigenstate. This process has many applications in the
engineering of quantum systems
\cite{Bergmann98,Wu13,Ribeiro13,Zhang13,Wang12}, and in particular plays
the fundamental role in adiabatic quantum computation (AQC)
\cite{Farhi00,Farhi01,Sarandy05}. The adiabatic theorem
\cite{Born28,Kato50} states that a system will undergo adiabatic evolution
given that the system parameter varies slowly.

Quantifying the applicability of adiabatic approximations is an interesting topic of current research efforts. On the one hand, this kind of research has been spurred by so-called shortcuts to adiabaticity \cite{Torrontegui13}, and on the other hand recent insights from thermodynamics haven put adiabatic processes back into focus \cite{Acconcia15,Sivak12}. In particular, the validity of the adiabatic theorem has been under intensive studies both theoretically and experimentally since it was proposed, and much of these efforts were devoted to the rigorous
description of the sufficient quantitative conditions of adiabatic theorem, and the estimation of the error accumulated over a
long time \cite{Kato50,Nenciu93,Avron12,Cao12}. Once the exact knowledge on the
adiabatic process is available, it is straightforward to apply the results to the optimal
design of adiabatic control on specific systems \cite{Rezakhani09,Wilson12}. The most interesting
progress is that the validity of the adiabatic theorem itself has been challenged in
the recent decade \cite{Marzlin04,Tong05,Tong07,Du08,Zhao08,Comparat09,Ortigoso12,Rigolin12,Zhang14}, both by strict
analysis and counter-examples. According to these findings, the errors induced by the adiabatic
approximation could accumulate over time despite certain quantitative condition is satisfied \cite{Marzlin04,Tong05,Tong07,Comparat09,Ortigoso12}, e.g., when
there exists an additional perturbation or driving that is resonant with the system. Particularly as indicated in \cite{Comparat09}, it is not new that resonant driving can cause population transfer between eigenstates. Also, a proof can be found in \cite{Ortigoso12} stating that only a resonant perturbation whose amplitude gradually decays to
zero can result in a violation of a well-known sufficient condition.

In this paper we consider the following process: the process starts at $t=0$. The system Hamiltonian at $t=0$
is $H_1$, and the system Hamiltonian at $t=T$ is $H_2=H_1+\lambda\Delta H,\lambda>0$. $\lambda$ is a dimensionless quantity. $\Delta H$ is a fixed operator and so the direction of the variation is fixed. We assume $H_1$, $H_2$, and $\Delta H$ are bounded operators throughout this
paper. $T$ is the evolution time. The transition of the system from $H_1$ to $H_2$ can be described using an interpolating function $f(t)$ so that
\begin{eqnarray}\label{ad1}
H(t)=H_1+f(t)(H_2-H_1)=H_1+\lambda f(t)\Delta H,
\end{eqnarray}
with $f(0)=0$ and $f(T)=1$. We work under the condition that a valid
perturbative analysis of the system evolution is available. This often
means $\lambda$ should be smaller than a threshold value. It is worth mentioning that the classical adiabatic theorem was proved also using a perturbative analysis, which cannot be applied directly to a large variation of Hamiltonian. Therefore, our
analysis in this paper is not concerned with the adiabatic evolution for a
large variation of Hamiltonian. However, our analysis provides a rigorous
estimation of the error accumulated during this
small-variation evolution for an
arbitrarily given interpolation.

Our work is different from the previous works in two ways. First, instead
of studying the evolution of the eigenstates and their corresponding
probability amplitudes, the mean value of a Hermitian operator is defined
as a measure of the error. For example, in the context of adiabatic quantum
computation where one wants to prepare the ground state of a target
Hamiltonian $\hat{H}_2\geq0$ whose ground-state energy is $0$, $\epsilon=\langle \hat{H}_2\rangle_{\rho_t}$ serves
as a good measure of the distance between the real-time state $\rho_t$ and the
ground state. This measure resembles the excitation energy or excess work performed during the process, as studied in thermodynamics \cite{Acconcia15}. In this paper we only consider the error accumulated over the
entire process, which means we are only interested in $\epsilon=\langle
\hat{H}_2\rangle_{\rho_T}$. The second difference is that the error, or the excitation energy or excess work performed during the process, can be estimated
with a sufficient precision for arbitrarily given interpolating
functions. As a result, the parameters which are related to the suppression
of the error can be easily identified. For example, we have
$\epsilon=O(\frac{\lambda^2}{T^2\lambda_2^3})$ as $\lambda\rightarrow0$ in
the case of linear
interpolation. Here $\lambda_2$ is the energy gap between the ground and first-excited states of the initial Hamiltonian. However for the interpolation in the counterexample
\cite{Marzlin04,Ortigoso12}, the scaling of $\epsilon$ is not so simple.

This paper is organized as follows. In Section \ref{secdef}, we introduce the model of this paper. In Section \ref{seclinear}, we give the estimation of the error for linear interpolation. In section \ref{secnonlinear}, we present the general algorithm to estimate the error for an arbitrarily given interpolating function. We discuss three examples in Section \ref{secexam}. Conclusion is given in section \ref{seccon}.

\section{Definitions and Preliminaries}\label{secdef}
The system is defined on an $N$-dimensional Hilbert space. We set Dirac constant $\hbar=1$. $||\cdot||$ denotes the matrix norm. Two real functions $f_1(x)$ and $f_2(x)$ can be denoted as $
f_1(x)=O(f_2(x)),\ x\rightarrow\infty,$ if and only if there exists a positive real number $M$ and a real number $x_0$ such that $|f_1(x)|\leq M|f_2(x)|,\ x\geq x_0$, where $|\cdot|$ denotes the absolute value.

Let $\{\omega_i:i=1,2,...N\}$ be the monotonically increasing sequence of
eigenvalues of $H_1$, so that $\omega_i\geq\omega_j$ when $i>j$, and
$\{|i\rangle\}$ be the corresponding eigenstates. We denote the energy gap
between the $i$th eigenstate and the ground state
as $\lambda_i=\omega_i-\omega_1$. Similarly, we define
the increasing sequence of eigenvalues of $H_2$, $\{\omega_i^{'}:i=1,2,...N\}$ and $\{\lambda_i^{'}\},\{|i^{'}\rangle\}$ correspondingly.

For convenience, we also introduce two offset Hamiltonians, $\hat{H}_1$ and
$\hat{H}_2$. The Hamiltonian $\hat{H}_1$ is defined as
$\hat{H}_1=H_1-\omega_1$, i.e., by offsetting the Hamiltonian of the system at
$t=0$ by a constant operator $\omega_1$ so that
  $\hat{H}_1\geq0$.
By $\hat{H}_1\geq0$ we mean $\hat{H}_1$ is positive semidefinite and its
the smallest eigenvalue of $\hat{H}_1$ is zero. Similarly, we define
$\hat{H}_2=H_2-\omega^{'}_1\geq0$ by offsetting
the system Hamiltonian by a constant operator $\omega^{'}_1$. Let $\rho_t$ denote the system state at time $t$ and let $\rho_g$
be the initial state of the system at $t=0$. We
always assume that $\rho_g$ is the ground state of $\hat{H}_1$, and so we
have $\langle \hat{H}_1\rangle_{\rho_g}=0$.

The measure of adiabaticity is proposed as follows
\begin{Def}
The distance between the final state and the ground state of $H_2$ is measured by
\begin{equation}\label{nde}
\epsilon=\langle \hat{H}_2\rangle_{\rho_T}.
\end{equation}
\end{Def}
Obviously, if the evolution is adiabatic, i.e., $\rho_T$ is the ground
state of $H_2$, then we have $\epsilon=0$. In particular, $\epsilon$ is closely related to the fidelity of the final state and ground state in the Schr{\"o}dinger picture (See Appendix \ref{appendix3}). A small error $\epsilon$ implies a large fidelity.

In this paper we also call $\epsilon$ the adiabatic approximation error, as $\epsilon$ reflects how well we can approximate the evolution as a perfect adiabatic process.

In this paper we only consider $\lambda$ such that  $\rho_t,\ t\in[0,T]$ can be expanded using Magnus series in the interaction picture.
For more details about the expansion in the interaction picture, please
refer to Appendix \ref{appendix1}. If the series expansion is valid in the
interaction picture, we can transform back to the Schr\"{o}dinger picture
and write the evolution of the state as (see Appendix \ref{appendix1})
\begin{eqnarray}\label{new1}
\rho_t=e^{-\mbox iH_1t}(\rho_g+R(t)+\mbox{i}[\rho_g,\lambda\int_0^tdt^{'}e^{\mbox iH_1t^{'}}f(t^{'})\Delta He^{-iH_1t^{'}}])e^{\mbox{i}H_1t},
\end{eqnarray}
where we have $||R(t)||=O(\lambda^2)$. A sufficient condition for the Magnus series to converge is given by (see Appendix \ref{appendix1})
\begin{equation}
\lambda<\frac{\pi}{\|\Delta H\|\int_0^Tf(t)dt}.
\end{equation}

Our aim is to estimate an asymptotic behaviour of $\epsilon$ provided
$\lambda\to 0$. Furthermore, we will use the obtained estimate
to analyze several cases of the adiabatic theorem including those where
some difficulties with adiabatic approximation have been encountered.

\section{Adiabatic approximation under linear interpolation of the Hamiltonian}\label{seclinear}
The Heisenberg evolution of the expectation of an observable is written as
\begin{equation}\label{heidy}
\frac{d}{dt}\langle X(t)\rangle_{\rho_g}=\langle-\mbox{i}[X(t),H]\rangle_{\rho_g},
\end{equation}
where $H$ is the system Hamiltonian. Recall that $\rho_g=|1\rangle\langle 1|$. Since $H_1|1\rangle=\omega_1|1\rangle$, $\langle X(t)\rangle_{\rho_g}$ is a constant of motion under the action of $H_1$:
\begin{equation}\label{ini}
\frac{d}{dt}\langle X(t)\rangle_{\rho_g}=\langle-\mbox{i}[X(t),H_1]\rangle_{\rho_g}=0=\langle-\mbox{i}[X(t),\hat{H}_1]\rangle_{\rho_g}
\end{equation}
for any Hermitian operator $X(t)$.

We will need to study the dynamics of $\langle \hat{H}_2\rangle_{\rho_t}=\langle \hat{H}_2(t)\rangle_{\rho_g}$ in order to solve for $\epsilon$. The time evolution of $\langle \hat{H}_2\rangle_{\rho_t}$ is determined by its generator $\frac{d}{dt}\langle \hat{H}_2\rangle_{\rho_t}=\langle-\mbox{i}[\hat{H}_2,H(t)]\rangle_{\rho_{t}}$. For linear interpolating function $f(t)=\frac{t}{T}$, integration of $\frac{d}{dt}\langle \hat{H}_2\rangle_{\rho_t}$ over $[0,T]$ results in the following expression (See details in Appendix \ref{appendix2}):
\begin{eqnarray}
&&\langle \hat{H}_2\rangle_{\rho_T}-\langle \hat{H}_2\rangle_{\rho_g}=\int_0^T(\langle-\mbox{i}[\hat{H}_2,H(t)]\rangle_{\rho_{t}})dt\nonumber\\
&=&\int_0^Tdt[-2(1-f(t))\sum_{i\neq 1}(\omega_i-\omega_1)\langle 1|\hat{H}_2|i\rangle\langle i|\hat{H}_2|1\rangle\int_0^tdt^{'}f(t^{'})\cos((\omega_i-\omega_1)(t^{'}-t))]\nonumber\\
&+&\int_0^Tdt\tr\{-\mbox{i}e^{\mbox{i}H_1t}[\hat{H}_2,(1-f(t))\hat{H}_1]e^{-\mbox{i}H_1t}R(t)\}.\label{ft}\\
&=&\sum_{i\neq 1}(-\frac{1}{\lambda_i}+\frac{4\sin^2(\lambda_iT/2)}{T^2\lambda_i^3})\langle 1|\hat{H}_2|i\rangle\langle i|\hat{H}_2|1\rangle\nonumber\\
&+&\int_0^Tdt\tr\{-\mbox{i}e^{\mbox{i}H_1t}[\hat{H}_2,(1-\frac{t}{T})\hat{H}_1]e^{-\mbox{i}H_1t}R(t)\}\label{ad9}
\end{eqnarray}
As we noted before, $\langle \hat{H}_2\rangle_{\rho_T}$ is exactly zero
  if $\rho_T$ is the ground state of $\hat{H}_2$. If $\rho_T$ is not the
  ground state of $\hat{H}_2$, we can determine the bound on $\epsilon=\langle
\hat{H}_2\rangle_{\rho_T}$ from the following equality
\begin{eqnarray}\label{int1}
\langle \hat{H}_2\rangle_{\rho_T}-\langle \hat{H}_2\rangle_{\rho_g}&=&\int_0^{T}\langle-\mbox{i}[\hat{H}_2,H(t)]\rangle_{\rho_{t}}dt\nonumber\\
&=&\int_0^{T}\langle-\mbox{i}[H_2,H(t)]\rangle_{\rho_{t}}dt=\langle H_2\rangle_{\rho_T}-\langle H_2\rangle_{\rho_g}.
\end{eqnarray}
Since
\begin{equation}
\hat{H}_2=H_2-\omega_1^{'},
\end{equation}
The error $\epsilon$ can be expressed as
\begin{equation}\label{int3}
\epsilon=\langle H_2\rangle_{\rho_T}-\langle H_2\rangle_{\rho_g}-[\omega_1^{'}-\langle H_2\rangle_{\rho_g}].
\end{equation}
With the aid of (\ref{ad9}), we can investigate the rate of convergence of $\epsilon$ to zero as $\lambda$ tends to zero in the
case where $f(t)$ defines a linear interpolation, as summarized in the
following proposition:
\begin{prop}\label{theorem1}
Assume $\lambda_2>0$ (the ground state of $H_1$ is non-degenerate) and
suppose $f(t)=t/T$, which corresponds to the linear interpolation of the
Hamiltonian. The estimation of $\epsilon$ is given by $\sum_{i\neq 1}\frac{4\lambda^2\sin^2(\lambda_iT/2)|\langle 1|\Delta H|i\rangle|^2}{T^2\lambda_i^3}+O(\lambda^3)$, which is of the order
$O(\frac{\lambda^2}{T^2\lambda_2^3})$ as $\lambda\rightarrow0$.
\end{prop}
\begin{proof}
Referring to (\ref{int3}) and (\ref{int1}), we need to compute the difference between (\ref{ad9}) and $\omega_1^{'}-\langle 1|H_2|1\rangle$. First we write (\ref{ad9}) as
\begin{eqnarray}
\langle H_2\rangle_{\rho_T}-\langle H_2\rangle_{\rho_g}&=&\sum_{i\neq 1}(-\frac{1}{\lambda_i}+\frac{4\sin^2(\lambda_iT/2)}{T^2\lambda_i^3})\langle 1|\hat{H}_2|i\rangle\langle i|\hat{H}_2| 1\rangle\label{aad1}\\
&+&\int_0^Tdt\tr\{-\mbox{i}e^{\mbox{i}H_1t}[\hat{H}_2,(1-f(t))\hat{H}_1]e^{-\mbox{i}H_1t}R(t)\}\nonumber\\
&=&-\sum_{i\neq 1}\frac{\lambda^2|\langle 1|\Delta H|i\rangle|^2}{\lambda_i}+\sum_{i\neq 1}\frac{4\lambda^2\sin^2(\lambda_iT/2)|\langle 1|\Delta H|i\rangle|^2}{T^2\lambda_i^3}\nonumber\\
&+&\int_0^Tdt\tr\{-\mbox{i}e^{\mbox{i}H_1t}[\hat{H}_2,(1-f(t))\hat{H}_1]e^{-\mbox{i}H_1t}R(t)\},\label{ad10}
\end{eqnarray}
by noting that
\begin{eqnarray}\label{ad11}
\sum_{i\neq 1}\frac{|\langle 1|\hat{H}_2|i\rangle|^2}{\lambda_i}=\sum_{i\neq 1}\frac{|\langle 1|H_1+\lambda\Delta H-\omega^{'}_1|i\rangle|^2}{\lambda_i}=\sum_{i\neq 1}\frac{\lambda^2|\langle 1|\Delta H|i\rangle|^2}{\lambda_i}.
\end{eqnarray}
Moreover, by the definition of the notation $O(\cdot)$ in Section~\ref{secdef} we can write $\sum_{i\neq 1}\frac{4\lambda^2\sin^2(\lambda_iT/2)|\langle 1|\Delta H|i\rangle|^2}{T^2\lambda_i^3}=O(\frac{\lambda^2}{T^2\lambda_2^3})$.

Denote $\bar{H}=\max_{f(t)\in(0,1)}||H(t)||$. Since
\begin{eqnarray}\label{ad12}
||\int_0^Tdt\tr\{-\mbox{i}e^{\mbox{i}H_1t}[\hat{H}_2,(1-f(t))\hat{H}_1]e^{-\mbox{i}H_1t}R(t)\}||\leq\frac{T\lambda}{2}\bar{H}^2||R(t)||
\end{eqnarray}
is $O(\lambda^3)$, we can further write (\ref{ad10}) as
\begin{equation}\label{twoerror}
\langle H_2\rangle_{\rho_T}-\langle H_2\rangle_{\rho_g}=-\sum_{i\neq 1}\frac{\lambda^2|\langle 1|\Delta H|i\rangle|^2}{\lambda_i}+\sum_{i\neq 1}\frac{4\lambda^2\sin^2(\lambda_iT/2)|\langle 1|\Delta H|i\rangle|^2}{T^2\lambda_i^3}+O(\lambda^3).
\end{equation}
Next we will calculate $\omega_1^{'}-\langle 1|H_2|1\rangle$. We have
\begin{equation}\label{ad13}
\langle 1|H_2|1\rangle=\langle 1|H_1+\lambda\Delta H|1\rangle=\omega_1+\lambda\langle 1|\Delta H|1\rangle.
\end{equation}
The smallest eigenvalue $\omega_1^{'}$ of $H_2$ can be calculated using the first-order time-independent perturbation theory for non-degenerate system. Assume $H_1$ is the unperturbed Hamiltonian and the perturbation is $\lambda\Delta H$, then the lowest eigenvalue of the perturbed Hamiltonian $H_1+\lambda\Delta H$ can be written as series in terms of $\lambda$ and $\omega_1$ \cite{Griffiths95}:
\begin{equation}\label{ad14}
\omega_1^{'}=\omega_1+\lambda\langle 1|\Delta H|1\rangle-\lambda^2\sum_{i\neq 1}\frac{|\langle 1|\Delta H|i\rangle|^2}{\lambda_i}+O(\lambda^3).
\end{equation}
Thus we conclude
\begin{equation}\label{ad15}
\omega_1^{'}-\langle H_2\rangle_{\rho_g}=\omega_1^{'}-\langle 1|H_2|1\rangle=-\lambda^2\sum_{i\neq1}\frac{|\langle1|\Delta H|i\rangle|^2}{\lambda_i}+O(\lambda^3).
\end{equation}
Comparing (\ref{twoerror}) and (\ref{ad15}), the terms $-\lambda^2\sum_{i\neq1}\frac{|\langle1|\Delta H|i\rangle|^2}{\lambda_i}$ cancel and so the error $\epsilon$ is estimated by
\begin{eqnarray}\label{prop1error}
\epsilon&=&\sum_{i\neq 1}\frac{4\lambda^2\sin^2(\lambda_iT/2)|\langle 1|\Delta H|i\rangle|^2}{T^2\lambda_i^3}+O(\lambda^3)\nonumber\\
&=&O(\frac{\lambda^2}{T^2\lambda_2^3}),\quad \lambda\rightarrow0.
\end{eqnarray}
\end{proof}

\section{Error Estimation for Arbitrary Interpolations}\label{secnonlinear}
The approach derived in the previous section can be easily generalized for
arbitrary given continuous interpolating functions. The generalization can
simply be done by replacing the linear interpolation function with the
given continuous function $f(t)$ and then recalculating the double
integration
\begin{equation}
A_i(T)=-2\int_0^Tdt\int_0^tdt^{'}(1-\frac{t}{T})\lambda_if(t^{'})\cos(\lambda_i(t^{'}-t))
\end{equation}
in (\ref{ft}). The error estimation can easily be obtained from the proof of Proposition \ref{theorem1}:
\begin{prop}\label{theorem2}
For an arbitrarily given $f(t)$, the error estimation is given by
\begin{equation}\label{thm2e}
\epsilon=\lambda^2\sum_{i\neq1}A_i(T)|\langle1|\Delta H|i\rangle|^2+\lambda^2\sum_{i\neq1}\frac{|\langle1|\Delta H|i\rangle|^2}{\lambda_i}+O(\lambda^3)
\end{equation}
as $\lambda\rightarrow0$.
\end{prop}
\begin{proof}
$\epsilon$ is still calculated by (\ref{int3}), using $\langle H_2\rangle_{\rho_T}-\langle H_2\rangle_{\rho_g}$ and $\omega_1^{'}-\langle H_2\rangle_{\rho_g}$. We have
\begin{equation}
\langle H_2\rangle_{\rho_T}-\langle H_2\rangle_{\rho_g}=\sum_{i\neq1}A_i(T)\lambda^2|\langle1|\Delta H|i\rangle|^2+O(\lambda^3)
\end{equation}
and
\begin{equation}
\omega_1^{'}-\langle H_2\rangle_{\rho_g}=-\lambda^2\sum_{i\neq1}\frac{|\langle1|\Delta H|i\rangle|^2}{\lambda_i}+O(\lambda^3).
\end{equation}
\end{proof}
It must be pointed out that $A(T)$ is very easy to calculate with the aid of any softwares that can perform symbolic integration, and therefore it is straightforward to apply Proposition \ref{theorem2} to find the error estimation for a given interpolating function, as we are going to do in the next section.

\section{Examples}\label{secexam}

\subsection{Linear Interpolation:$\ f(t)=t/T$}
By Proposition \ref{theorem1}, the error estimation is $\epsilon=\sum_{i\neq 1}\frac{4\sin^2(\lambda_iT/2)}{T^2\lambda_i^3}|\langle 1|\Delta H|i\rangle|^2\lambda^2+O(\lambda^3)$ as $\lambda\rightarrow0$. Since $\sin^2(\lambda_iT/2)$ and $\Delta H$ are bounded, this error term is primarily determined by $\frac{\lambda}{T}$ which is the average speed of the variation of the system Hamiltonian, and $\frac{1}{\lambda_i}$ which is the inverse of the energy gap between the ground and $i$-th eigenstates of $H_1$, as $\lambda\rightarrow0$. In particular, we have
\begin{equation}
\lim_{{\lambda}\rightarrow0}\frac{\epsilon}{(\frac{\lambda}{T})^2}=\sum_{i\neq 1}\frac{4\sin^2(\lambda_iT/2)}{\lambda_i^3}|\langle 1|\Delta H|i\rangle|^2.
\label{VU.lininterp}
\end{equation}
Therefore, when the inverse of the energy gaps $\frac{1}{\lambda_i}$ are fixed values, the approximation error $\epsilon$ is estimated to be proportional to the square of the average speed of the variation of the Hamiltonian, which is $(\frac{\lambda}{T})^2$, as $\lambda\rightarrow0$.

\subsection{Quadratic Interpolation:$\ f(t)=t^2/T^2$}
Replace $f(t)$ with a nonlinear function $f(t)=\frac{t^2}{T^2}$ in
(\ref{ft}) and we recalculate the integral to be
\begin{eqnarray}\label{nonlinear1}
\sum_{i\neq1}A_i(T)=\sum_{i\neq 1}(-\frac{1}{\lambda_i}+\frac{16\sin^2(\frac{T\lambda_i}{2})+4T^2\lambda_i^2-8T\lambda_i\sin(T\lambda_i)}{T^4\lambda_i^5})|\langle1|\Delta H|i\rangle|^2.\nonumber\\
\end{eqnarray}
By Proposition \ref{theorem2}, for sufficiently small $\lambda$, the error is estimated to be of order of $\lambda^2$:
\begin{eqnarray}
\epsilon_{quad}&=&\lambda^2\sum_{i\neq 1}\frac{16\sin^2(\frac{T\lambda_i}{2})+4T^2\lambda_i^2-8T\lambda_i\sin(T\lambda_i)}{T^4\lambda_i^5}|\langle1|\Delta H|i\rangle|^2+O(\lambda^3)\nonumber\\
&=&(\frac{\lambda}{T})^2\sum_{i\neq
  1}[\frac{16\sin^2(\frac{T\lambda_i}{2})}{T^2\lambda_i^5}+\frac{4}{\lambda_i^3}-\frac{8\sin(T\lambda_i)}{T\lambda_i^4}]|\langle1|\Delta H|i\rangle|^2+O(\lambda^3).\nonumber\\
\label{VO.quadinterp}
\end{eqnarray}
That is, in contrast to the linear interpolation case, we have
\begin{equation}
\lim_{\lambda\rightarrow0}\frac{\epsilon_{quad}}{(\frac{\lambda}{T})^2}=\sum_{i\neq
  1}[\frac{16\sin^2(\frac{T\lambda_i}{2})}{T^2\lambda_i^5}+\frac{4}{\lambda_i^3}-\frac{8\sin(T\lambda_i)}{T\lambda_i^4}]|\langle1|\Delta H|i\rangle|^2.
\end{equation}
This calculation shows that if the evolution speed is infinitely slow, then
the system dynamics is adiabatic during $t\in[0,T]$. However, the scaling
of $\epsilon_{quad}$ with respect of the square of the average evolution speed $\frac{\lambda}{T}$ is not as simple
as in the linear case, where the scaling of $\epsilon$ with respect of $(\frac{\lambda}{T})^2$ is primarily determined by the inverse of the energy gaps as $\lambda\rightarrow0$.
In the quadratic case, this scaling is primarily determined by a complex factor
$[\frac{16\sin^2(\frac{T\lambda_i}{2})}{T^2\lambda_i^5}+\frac{4}{\lambda_i^3}-\frac{8\sin(T\lambda_i)}{T\lambda_i^4}]$
which depends mainly on the inverse of the energy gaps $\{\lambda_i\}$ and the inverse of the evolution time $T$.

\subsection{Interpolation with Decaying Resonant Terms}
Here we assume a linear interpolating function with an additional
oscillating term that gradually decays to zero. That is,
\[
f(t)=\frac{t}{T}+g(1-\frac{t}{T})\sin(\lambda_ct),
\]
where $\lambda_c$ is the oscillating frequency of the
perturbation. Ortigoso observed in \cite{Ortigoso12} the inconsistency
in the applicability of the adiabatic theorem when the Hamiltonian contains
resonant terms whose amplitudes go asymptotically to zero.

Replace $f(t)$ with $f(t)=\frac{t}{T}+g(1-\frac{t}{T})\sin(\lambda_ct)$ in
(\ref{ft}) and we recalculate the integral to be
\begin{equation}\label{nonlinear4}
\sum_{i\neq1}A_i(T)=\sum_{i\neq 1}\frac{Q_1(g,T,\lambda_i,\lambda_c)}{T^2(2\lambda_i^{11}-8\lambda_i^9\lambda_c^2+12\lambda_i^7\lambda_c^4-8\lambda_i^5\lambda_c^6+2\lambda_i^3\lambda_c^8)}.
\end{equation}
$Q_1$ is a function of four parameters. In particular, we note that each term in (\ref{nonlinear4}) is well defined for all $\lambda_c$, including $\lambda_c=\lambda_i$, since as $\lambda_c\to\lambda_i$, the $i$-th term in (\ref{nonlinear4}) approaches
\begin{eqnarray}\label{nonlinear5}
&&[-128\sin^2(\frac{T\lambda_i}{2})-16g\sin(T\lambda_i)+8g\sin(2T\lambda_i)+g^2\nonumber\\
&+&g^2(2\sin^2(T\lambda_i)-1)-32T^2\lambda^2-16gT\lambda_i+2g^2T^4\lambda_i^4\nonumber\\
&+&16gT^2\lambda_i^2\sin(T\lambda_i)-16gT\lambda_i(2\sin^2(\frac{T\lambda_i}{2})-1)\nonumber\\
&-&2g^2T^2\lambda_i^2(2\sin^2(T\lambda_i)-1)-2g^2T\lambda_i\sin(2T\lambda_i)]/32T^2\lambda_i^3\nonumber\\
&=&-\frac{1}{\lambda_i}+\frac{g^2}{16}T^2\lambda_i+\frac{g}{2\lambda_i}\sin(T\lambda_i)-\frac{g^2}{16\lambda_i}(2\sin^2(T\lambda_i)-1)+Q(T),
\end{eqnarray}
where $Q(T)$ is a complicated fraction with $T$ being in its denominator. The error resulting
from the $i$-th term is given by
\begin{eqnarray}\label{counte}
&&\epsilon_i\nonumber\\
&=&|\langle1|\Delta H|i\rangle|^2[\frac{g^2T^4\lambda_i}{16}+\frac{gT^2\sin(T\lambda_i)}{2\lambda_i}-\frac{g^2T^2(2\sin^2(T\lambda_i)-1)}{16\lambda_i}+T^2Q(T)](\frac{\lambda}{T})^2\nonumber\\
&+&O(\lambda^3)
\end{eqnarray}
as $\lambda\rightarrow0$. We have
\begin{equation}
\lim_{\lambda\rightarrow0}\frac{\epsilon_i}{(\frac{\lambda}{T})^2}=|\langle1|\Delta H|i\rangle|^2[\frac{g^2T^4\lambda_i}{16}+\frac{gT^2\sin(T\lambda_i)}{2\lambda_i}-\frac{g^2T^2(2\sin^2(T\lambda_i)-1)}{16\lambda_i}+T^2Q(T)].
\label{VO.oscinterp}
\end{equation}
The scaling of $\epsilon_i$ with respect of $(\frac{\lambda}{T})^2$ is additionally determined by $T^2$ and $T^4$, as compared to the quadratic case. This is
where adiabatic approximation error may not be small if the average evolution speed is slow. In particular by (\ref{VO.oscinterp}), if one chooses a comparably large value for $T$ in an adiabatic evolution experiment,
the adiabatic approximation error may not decrease as expected when one applies a slow evolution speed $\frac{\lambda}{T}$.

In order to further illustrate this point, we can heuristically compare the speed of convergence of $\epsilon$ to zero observed in this case and the quadratic case, as the speed of the adiabatic process ($\lambda/T$) reduces and the evolution horizon $T$ increases. The difference in the speed of convergence can be clearly seen using the ratio
\begin{equation}\label{ratio}
\lim_{T\to\infty}\left(\lim_{(\lambda/T)\to 0}\frac{\epsilon_i}{\epsilon_{quad}}\right)=\infty.
\end{equation}
Therefore, the rate of convergence considered in this subsection is slower than that in the quadratic
or linear case. i.e., $\epsilon$ goes to zero as $\lambda\to 0$ at a much slower
rate than in the linear interpolation case or the quadratic interpolation case
if $T$ is large. Furthermore, the larger $T$ is, the slower
the convergence.

\section{Conclusion}\label{seccon}
In this paper we provide a rigorous analysis of the time-dependent
evolution of Hamiltonian-varying quantum systems. As we calculated, the adiabatic approximation error
is not proportional to the average speed of the variation of the system Hamiltonian and the inverse
of the energy gaps in many cases. The results in this paper may provide guidelines when applying complicated interpolation for adiabatic evolution.

\appendix
\section{}\label{appendix1}
The Magnus expansion is proposed to solve the following time-dependent equation \cite{Blanes09}
\begin{equation}
\frac{dY(t)}{dt}=A(t)Y(t).
\end{equation}
The solution of the above equation can be written as
\begin{equation}
Y(t)=\exp(\sum_{k=1}^{\infty}\Omega_k(t))Y(0),
\end{equation}
where the first three terms in the Magnus series $\{\Omega_k,k=1,2,\cdot\cdot\cdot,\infty\}$ are calculated by
\begin{eqnarray}
\Omega_1(t)&=&\int_0^tA(t_1)dt_1,\nonumber\\
\Omega_2(t)&=&\frac{1}{2}\int_0^tdt_1\int_0^{t_1}dt_2[A(t_1),A(t_2)],\nonumber\\
\Omega_3(t)&=&\frac{1}{6}\int_0^tdt_1\int_0^{t_1}dt_2\int_0^{t_2}dt_3([A(t_1),[A(t_2),A(t_3)]]+[A(t_3),[A(t_2),A(t_1)]]).\nonumber\\
\end{eqnarray}
The rest terms in the Magnus series can also be written as the integrals of nested commutators.

The dynamical equation of the quantum state in interaction picture is given by
\begin{equation}\label{ad2}
\mbox{i}\frac{\partial|\psi_I(t)\rangle}{\partial t}=e^{\mbox{i}H_1t}\lambda f(t)\Delta He^{-\mbox{i}H_1t}|\psi_I(t)\rangle,
\end{equation}
where $|\psi(t)\rangle=e^{-\mbox{i}H_1t}|\psi_I(t)\rangle$. Applying the Magnus expansion to (\ref{ad2}) yields
\begin{eqnarray}
|\psi_I(t)\rangle=(1-\mbox{i}\lambda\int_0^tdt^{'}e^{\mbox{i}H_1t^{'}}f(t^{'})\Delta He^{-\mbox{i}H_1t^{'}}+R_0(t))|\psi(0)\rangle,\label{ad3}
\end{eqnarray}
where $R_0(t)$ includes all the higher-order terms as determined by $\{\Omega_k\}$. Obviously, $||R_0(t)||$ is of the order $O(\lambda^2)$. Transforming back to Schr{\"o}dinger picture we obtain the expression for the density operator as
\begin{eqnarray}\label{ad5}
\rho_t&=&e^{-\mbox{i}H_1t}(1-\mbox{i}\lambda\int_0^tdt^{'}e^{\mbox{i}H_1t^{'}}f(t^{'})\Delta He^{-\mbox{i}H_1t^{'}}\nonumber\\
&+&R_0(t))\rho_g(1+\mbox{i}\lambda\int_0^tdt^{'}e^{\mbox{i}H_1t^{'}}f(t^{'})\Delta He^{-\mbox{i}H_1t^{'}}+R_0^\dagger(t))e^{\mbox{i}H_1t}\nonumber\\
&=&e^{-\mbox{i}H_1t}(\rho_g+R(t)+\mbox{i}[\rho_g,\lambda\int_0^tdt^{'}e^{\mbox{i}H_1t^{'}}f(t^{'})\Delta He^{-\mbox{i}H_1t^{'}}])e^{\mbox{i}H_1t}.
\end{eqnarray}
Obviously, $||R(t)||$ is also of the order $O(\lambda^2)$.

An explicit condition for the Magnus series to converge is given by \cite{Blanes09}
\begin{equation}\label{convergecon2}
||\Delta H||\int_0^T\lambda f(t)dt<\pi.
\end{equation}

\section{}\label{appendix2}
The derivative of $\langle \hat{H}_2\rangle_{\rho_t}$ is calculated as
\begin{eqnarray}\label{ad7}
&&\frac{d}{dt}\langle \hat{H}_2\rangle_{\rho_t}\nonumber\\
&=&\langle-\mbox{i}[\hat{H}_2,H_1+f(t)(H_2-H_1)]\rangle_{\rho_t}=\langle-\mbox{i}[\hat{H}_2,\hat{H}_1+f(t)(\hat{H}_2-\hat{H}_1)]\rangle_{\rho_t}\nonumber\\
&=&\langle-\mbox{i}e^{\mbox{i}H_1t}[\hat{H}_2,(1-f(t))\hat{H}_1]e^{-\mbox{i}H_1t}\rangle_{\rho_g+\mbox{i}[\rho_g,\lambda\int_0^tdt^{'}e^{\mbox{i}H_1t^{'}}f(t^{'})(\Delta H)e^{-\mbox{i}H_1t^{'}}]+R(t)}\nonumber\\
&=&-\langle[e^{\mbox{i}H_1t}[\hat{H}_2,(1-f(t))\hat{H}_1]e^{-\mbox{i}H_1t},\lambda\int_0^tdt^{'}e^{\mbox{i}H_1t^{'}}f(t^{'})\Delta He^{-\mbox{i}H_1t^{'}}]\rangle_{\rho_g}\nonumber\\
&+&\tr\{-\mbox{i}e^{\mbox{i}H_1t}[\hat{H}_2,(1-f(t))\hat{H}_1]e^{-\mbox{i}H_1t}R(t)\}\nonumber\\
&=&-\langle[e^{\mbox{i}H_1t}[\hat{H}_2,(1-f(t))\hat{H}_1]e^{-\mbox{i}H_1t},\int_0^tdt^{'}e^{\mbox{i}H_1t^{'}}f(t^{'})\hat{H}_2e^{-\mbox{i}H_1t^{'}}]\rangle_{\rho_g}\nonumber\\
&+&\tr\{-\mbox{i}e^{\mbox{i}H_1t}[\hat{H}_2,(1-f(t))\hat{H}_1]e^{-\mbox{i}H_1t}R(t)\},
\end{eqnarray}
where we made use of the relation (\ref{ini}) and $\lambda\Delta H=\hat{H}_2-\hat{H}_1+\omega^{'}_1-\omega_1$. Calculating (\ref{ad7}) further leads to
\begin{eqnarray}\label{ad8}
&&\frac{d}{dt}\langle \hat{H}_2\rangle_{\rho_t}=-(1-f(t))\langle e^{\mbox{i}\omega_1t}\hat{H}_2\hat{H}_1e^{-\mbox{i}H_1t}\int_0^tdt^{'}e^{\mbox{i}H_1t^{'}}f(t^{'})\hat{H}_2e^{-\mbox{i}\omega_1t^{'}}\nonumber\\
&+&\int_0^tdt^{'}e^{\mbox{i}\omega_1t^{'}}f(t^{'})\hat{H}_2e^{-\mbox{i}H_1t^{'}}e^{\mbox{i}H_1t}\hat{H}_1\hat{H}_2e^{-\mbox{i}\omega_1t}\rangle_{\rho_g}\nonumber\\
&+&\tr\{-\mbox{i}e^{\mbox{i}H_1t}[\hat{H}_2,(1-f(t))\hat{H}_1]e^{-\mbox{i}H_1t}R(t)\}\nonumber\\
&=&-(1-f(t))\langle e^{\mbox{i}\omega_1t}\hat{H}_2\hat{H}_1\sum|i\rangle\langle i|e^{-\mbox{i}H_1t}\int_0^tdt^{'}e^{\mbox{i}H_1t^{'}}f(t^{'})\hat{H}_2e^{-\mbox{i}\omega_1t^{'}}\nonumber\\
&+&\int_0^tdt^{'}e^{\mbox{i}\omega_1t^{'}}f(t^{'})\hat{H}_2e^{-\mbox{i}H_1t^{'}}e^{\mbox{i}H_1t}\sum|i\rangle\langle i|\hat{H}_1\hat{H}_2e^{-\mbox{i}\omega_1t}\rangle_{\rho_g}\nonumber\\
&+&\tr\{-ie^{\mbox{i}H_1t}[\hat{H}_2,(1-f(t))\hat{H}_1]e^{-\mbox{i}H_1t}R(t)\}\nonumber\\
&=&-(1-f(t))\sum_{i\neq 1}(\omega_i-\omega_1)\langle \hat{H}_2|i\rangle\langle i|\hat{H}_2\rangle_{\rho_g}\int_0^tdt^{'}f(t^{'})(e^{\mbox{i}(\omega_i-\omega_1)(t^{'}-t)}+e^{\mbox{i}(\omega_i-\omega_1)(t-t^{'})})\nonumber\\
&+&\tr\{-\mbox{i}e^{\mbox{i}H_1t}[\hat{H}_2,(1-f(t))\hat{H}_1]e^{-\mbox{i}H_1t}R(t)\}\nonumber\\
&=&-2(1-f(t))\sum_{i\neq 1}(\omega_i-\omega_1)\langle 1|\hat{H}_2|i\rangle\langle i|\hat{H}_2|1\rangle\int_0^tdt^{'}f(t^{'})\cos((\omega_i-\omega_1)(t^{'}-t))\nonumber\\
&+&\tr\{-\mbox{i}e^{\mbox{i}H_1t}[\hat{H}_2,(1-f(t))\hat{H}_1]e^{-\mbox{i}H_1t}R(t)\}.
\end{eqnarray}
With the linear interpolating function $f(t)=\frac{t}{T}$, the direct integration of (\ref{ad8}) over $[0,T]$ gives (\ref{ad9}).

\section{}\label{appendix3}
The state of the system will remain a pure state during the evolution. Therefore, we can express the final state as $\rho_T=|\psi\rangle\langle\psi|$ with $|\psi\rangle=\sum_{i=1}^Nc_i|i^{'}\rangle$. Using this expression, the error measure $\epsilon$ defined by (\ref{nde}) can be written as
\begin{equation}
\epsilon=\langle \hat{H}_2\rangle_{\rho_T}=\langle\psi|\hat{H}_2|\psi\rangle=\sum_{i=2}^N|c_i|^2\lambda_i^{'}\geq\sum_{i=2}^N|c_i|^2\lambda_2^{'}.
\end{equation}
The fidelity of the final state and the ground state $|1^{'}\rangle$ is calculated by
\begin{equation}
F(|\psi\rangle,|1^{'}\rangle)=\sqrt{|\langle\psi|1^{'}\rangle|^2}=\sqrt{|c_1|^2}=\sqrt{1-\sum_{i=2}^N|c_i|^2}\geq\sqrt{1-\frac{\epsilon}{\lambda_2^{'}}}.
\end{equation}

%%%%%%%%%%%%%%%%%%%%%%%%%%%%%%%%%%%%%%%%%%%%%%
%%                                          %%
%% Backmatter begins here                   %%
%%                                          %%
%%%%%%%%%%%%%%%%%%%%%%%%%%%%%%%%%%%%%%%%%%%%%%

\begin{backmatter}

\section*{Competing interests}
  The authors declare that they have no competing interests.

\section*{Acknowledgements}
Yu Pan would like to thank Li Li and Charles Hill for their valuable suggestions. We gratefully acknowledge support by the Australian Research Council Centre of Excellence for Quantum Computation
and Communication Technology (project number CE110001027), Australian
Research Council Discovery Project (projects DP110102322 and DP140101779).
%%  After submission of the .TEX file,                     %%
%%  you will be prompted to submit your .BBL file.         %%

% if your bibliography is in bibtex format, use those commands:
\bibliographystyle{bmc-mathphys} % Style BST file (bmc-mathphys, vancouver, spbasic).
%\bibliography{ref}      % Bibliography file (usually '*.bib' )

% or include bibliography directly:
% \begin{thebibliography}
% \bibitem{b1}
% \end{thebibliography}

%%%%%%%%%%%%%%%%%%%%%%%%%%%%%%%%%%%
%%                               %%
%% Figures                       %%
%%                               %%
%% NB: this is for captions and  %%
%% Titles. All graphics must be  %%
%% submitted separately and NOT  %%
%% included in the Tex document  %%
%%                               %%
%%%%%%%%%%%%%%%%%%%%%%%%%%%%%%%%%%%

%%
%% Do not use \listoffigures as most will included as separate files

%%%%%%%%%%%%%%%%%%%%%%%%%%%%%%%%%%%
%%                               %%
%% Tables                        %%
%%                               %%
%%%%%%%%%%%%%%%%%%%%%%%%%%%%%%%%%%%

%% Use of \listoftables is discouraged.
%%

%%%%%%%%%%%%%%%%%%%%%%%%%%%%%%%%%%%
%%                               %%
%% Additional Files              %%
%%                               %%
%%%%%%%%%%%%%%%%%%%%%%%%%%%%%%%%%%%

\end{backmatter}
\end{document}